\documentclass{article}

\usepackage{amsmath}
\usepackage[american]{babel}
\usepackage{graphicx}
\usepackage[utf8x]{inputenc}
\usepackage{dsfont}
\usepackage{subfigure}
\usepackage{todonotes}
\usepackage{wrapfig}
\usepackage{xspace}
\usepackage[pdfborder={0 0 0}]{hyperref} 

\graphicspath{{./graphics/}}

\newcommand{\N}{\mathds{N}}
\newcommand{\NP}{\textit{NP}}
\newcommand{\eg}{e.g.\@\xspace}
\newcommand{\ie}{i.e.\@\xspace}

\newcommand{\obda}{w.\,l.\,o.\,g.\@\xspace} 

\newcommand{\V}{\mathcal{V}}

\newtheorem{theorem}{Theorem}[section]
\newtheorem{lemma}[theorem]{Lemma}

\newenvironment{proof}[1][Proof]{\begin{trivlist}
\item[\hskip \labelsep {\bfseries #1}]}{\end{trivlist}}
\newenvironment{definition}[1][Definition]{\begin{trivlist}
\item[\hskip \labelsep {\bfseries #1}]}{\end{trivlist}}

\newcommand{\qed}{\nobreak \ifvmode \relax \else
      \ifdim\lastskip<1.5em \hskip-\lastskip
      \hskip1.5em plus0em minus0.5em \fi \nobreak
      \vrule height0.75em width0.5em depth0.25em\fi}

\makeatletter
\newcommand*{\etc}{%
    \@ifnextchar{.}%
        {etc}%
        {etc.\@\xspace}%
}
\makeatother

\newcommand{\ignore}[1]{}

\title{Complexity of the General Chromatic Art Gallery Problem}

\author{Sándor~P.~Fekete\thanks{Department of Computer Science, TU Braunschweig, Germany.
		{\tt s.fekete@tu-bs.de, stephan.friedrichs@tu-bs.de, mhsaar@gmail.com}}
	\and Stephan~Friedrichs*
	\and Michael~Hemmer*}

\begin{document}

\maketitle

\begin{abstract}
For a polygonal region $P$ with $n$ vertices, a {\em guard cover} $G$ is a set
of points in $P$, such that any point in $P$ can be seen from a point in $G$.
In a {\em colored guard cover}, every element in a guard cover is assigned a
color, such that no two guards with the same color have overlapping visibility
regions. The Chromatic Art Gallery Problem (CAGP) asks for the minimum number
of colors for which a colored guard cover exists.

We provide first complexity results for the general CAGP, in which arbitrary
guard positions are allowed. We show that it is already \NP-hard
to decide whether two colors suffice for covering a polygon with holes, as well
as for any fixed number $k\geq 2$ of colors. For simple polygons, we show
that computing the minimum number of colors is \NP-hard
for $\Theta(n)$ colors.
\end{abstract}

\section{Introduction}
\label{sec:introduction}

Consider a robot moving in a polygonal environment $P$. At any point $p$ in $P$,
the robot can navigate by referring to a beacon that is directly
visible from $p$. In order to ensure unique orientation, each beacon
has a ``color''; the same color may be used for different beacons, if their
visibility regions do not overlap. What is the minimum number $\chi_G(P)$ of colors
for covering $P$, and where should the beacons be placed?

This is the {\sc Chromatic Art Gallery Problem} (CAGP), which was first introduced by
Erickson and LaValle~\cite{el-cagp-10}.
Clearly, any feasible set of beacons for the CAGP must also be a feasible solution for the classical
{\sc Art Gallery Problem} (AGP). 
However, the number of guards and their positions for optimal AGP and CAGP solutions can be quite different, even in cases as simple as the one shown in Figure~\ref{fig:example}.

\begin{figure}
  \begin{center}
    \def\svgwidth{.4\textwidth}
    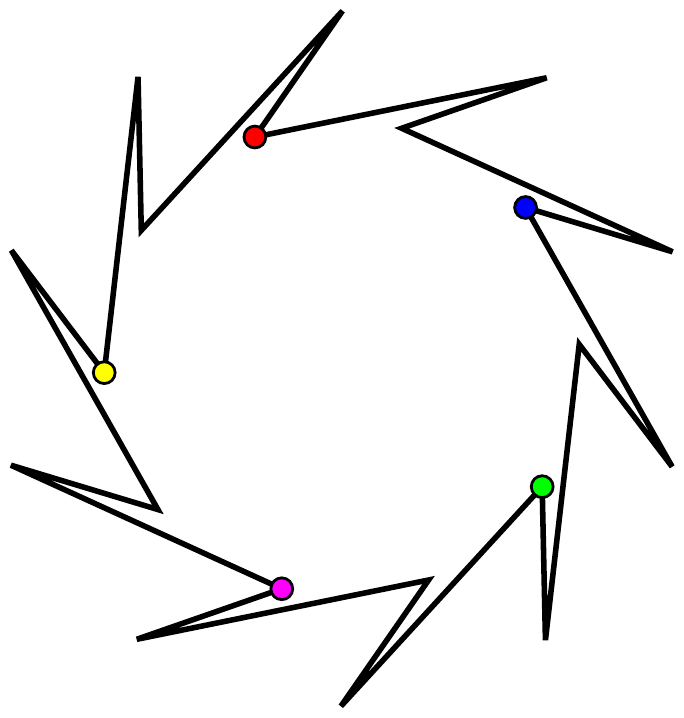
    \def\svgwidth{.4\textwidth}
    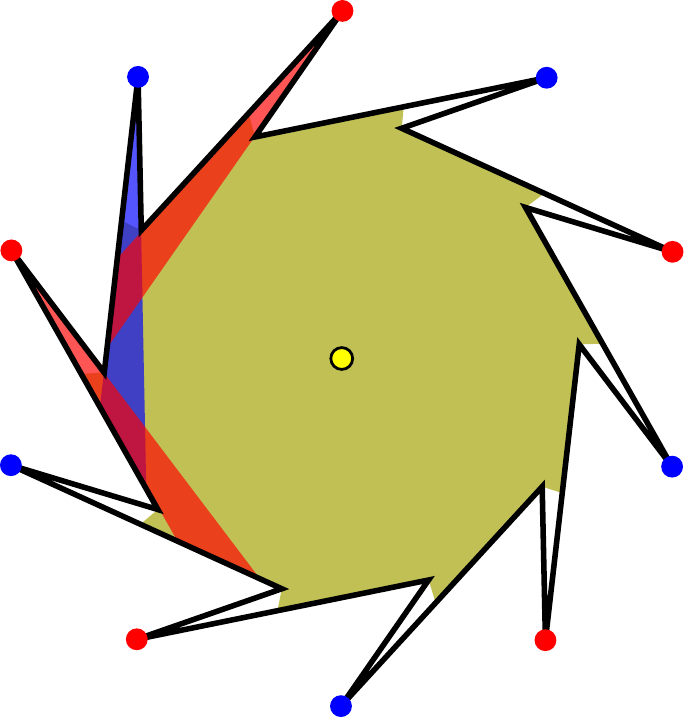
  \end{center}
  \caption{%
    An example polygon with $n=20$ vertices. 
     A minimum-cardinality guard cover with $n/4$ guards  requires $n/4$ colors (left) while 
     a minimum-color guard cover has $n/2+1$ guards and requires only $3$ colors (right).}
  \label{fig:example}
\end{figure}

{\bf Related Work.}
The closely related AGP is \NP-hard~\cite{ll-ccagp-86}, even for simple polygons. 
See~\cite{r-agta-87,s-rrag-92,u-agip-00} for three surveys with a wide variety of results.
More recently, there has been work on developing practical optimization methods for computing optimal AGP solutions~\cite{kbf+-esbga-12,trs-qosag-13,ffk+-fagp-13b,bdd-pgpc-13}.

The CAGP was first proposed by Erickson and LaValle, who presented a number of results, most notably upper and lower bounds on the number of colors for different classes of polygons~\cite{el-cagp-10,el-agael-11,el-hmlcn-11}.
In particular, they noted that the construction of Lee and Lin~\cite{ll-ccagp-86} which establishes \NP-hardness of determining a minimum-cardinality guard cover for a simple polygon $P$ can be used to prove \NP-hardness of computing $\chi_G(P)$, as long as all guards have to be picked from a specified candidate set.
However, there is no straightforward way to extend this construction for showing \NP-hardness of the CAGP with {\em arbitrary} guard positions.
An upcoming paper by Zambon et al.~\cite{zrs-aeaftcagp-14} discusses worst-case bounds, as well as exact methods for computing optimal solutions.

Bärtschi and Suri introduced the \textsc{Conflict-Free CAGP}, in which they relax the chromatic requirements~\cite{bs-cfcag-12}:
Visibility regions of guards with the same color may overlap, as long as there is always one uniquely colored guard visible.
They showed an upper bound of $O(\log^2 n)$ on the worst-case number of guards; this is significantly smaller
than the lower bound of $\Theta(n)$ for $\chi_G(P)$ established by Erickson and LaValle~\cite{el-agael-11}.

Another loosely related line of research is by Biro et al.~\cite{bik+-bbagr-13}, who consider beacons of a different kind:
in order to get to a new location, a robot aims for a sequence of destinations for shortest geodesic paths, each from a finite set of beacons.
Their results include a tight worst-case bound of
$\left\lfloor\frac{n}{2}\right\rfloor-1$ for the number of beacons and a proof of \NP-hardness for finding a smallest set of beacons in a simple polygon.

{\bf Our Results.}
We show that the CAGP is \NP-hard for a number of scenarios, even if we are {\em not} restricted to a
fixed set of guard positions. In particular, we show the following:
\begin{itemize}
\item For a polygon with holes, it is \NP-hard to decide whether there is a $k$-colorable guard cover, for any fixed $k\geq 2$.
Because $k=1$ requires a single guard, i.e., a star-shaped polygon, we get a complete complexity analysis.
\item For a simple polygon, it is \NP-hard to compute an optimal solution for the CAGP. The proof establishes this for
$\chi_G(P)\in\Theta(n)$.
\end{itemize}

\section{Preliminaries}
\label{sec:preliminaries}

Let $P$ be a polygon.
$P$ is \emph{simple} if its boundary is connected.
For $p \in P$, $\V(p) \subseteq P$ denotes the \emph{visibility polygon} of $p$, \ie, all points $p'$ that can be connected to $p$ using the line segment $\overline{pp'} \subset P$.
For any $G \subseteq P$, we denote by $\V(G) = \bigcup_{g \in G} \V(g)$.
A finite $G \subset P$ with $\V(G) = P$ is called a \emph{guard cover} of $P$; $g \in G$ is a \emph{guard}.
We say that $g$ \emph{covers} all points in $\V(g)$.

Let $G$ be a guard cover of $P$ and let $c: G \to \{1, \dots, k\}$ be a coloring of the guards.
Then $(G, c)$ is a \emph{$k$-coloring} of $P$ if no point sees two guards of the same color.
$\chi_G(P)$ is the \emph{chromatic (guard) number} of $P$, i.e., the minimal $k$, such that there is a $k$-coloring of $P$.
(The index $G$ in $\chi_G(P)$ as introduced in~\cite{el-agael-11,el-hmlcn-11}
does not refer to a specific guard set.)

\begin{definition}[Chromatic Art Gallery Problem]\label{def:cagp}
For $k \in \N$, the \emph{$k$-Chromatic Art Gallery Problem ($k$-CAGP)} is the following decision problem:
Given a polygon $P$, decide whether $\chi_G(P) \leq k$.
\end{definition}

We provide a simple proof of the \NP-hardness of $k$-CAGP for $k \geq 3$ in Section~\ref{sec:three-colorability} and a more complex one for the \NP-hardness of $2$-CAGP in Section~\ref{sec:two-colorability}.

\begin{lemma}[Needle lemma]\label{lem:needles}
Consider $k+1$ needles with end points $W = \{ w_1, \dots, w_{k+1} \}$, 
such that there is some $g \in \bigcap_{i=1}^{k+1} \V(w_i)$ with $\V(W) \subseteq \V(g)$, 
as in Figure~\ref{fig:needles}.
Then a $k$-coloring must place a guard in $V = \{ p \in P \mid |\V(p) \cap W| \geq 2 \}$, \eg at $g$.
\end{lemma}

\begin{proof}
Suppose there is no guard in $V$.
Then there must be a guard in each $\V(w_i)$, requiring a 
total of at least $k+1$ guards; two of them share the same color.
Because $\V(W) \subseteq \V(g)$, $g$ sees two guards with the same 
color, which is impossible in a $k$-coloring.
\end{proof}

\begin{figure}
  \begin{center}
    \def\svgwidth{.5\textwidth}
    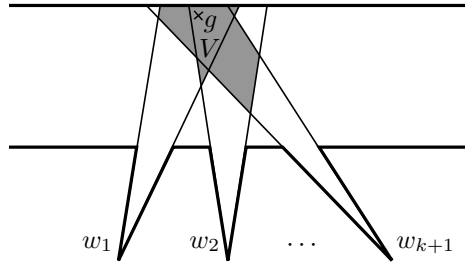
  \end{center}
  \caption{Forcing one guard into the region $V$ (gray).}
  \label{fig:needles}
\end{figure}

Note that for every $k \in \N$, there is a polygon that is not $k$-colorable; see
Erickson and LaValle~\cite{el-agael-11}.

%
%

\section{3- and (3+k)-Colorability}
\label{sec:three-colorability}

The \NP-hardness of $3$-CAGP follows from the \NP-hardness of deciding whether a planar graph $H$ is
3-colorable-complete~\cite{gj-cai-79}.
The idea is shown in continuous lines in Figure~\ref{fig:3color}:
Each node of the graph $H$ is turned into a convex region of the polygon $P$, and needles force exactly one guard into each of them by Lemma~\ref{lem:needles}.
Those guards cover $P$; it is easy to see that $P$ is $3$-colorable iff $H$ is $3$-colorable.
For $1 \leq k \in \N$, the construction can be generalized to $(3+k)$-CAGP by adding the structures drawn in dashed lines to Figure~\ref{fig:3color}.
We omit details for lack of space.

\begin{figure}
	\begin{center}
		\def\svgwidth{\linewidth}
		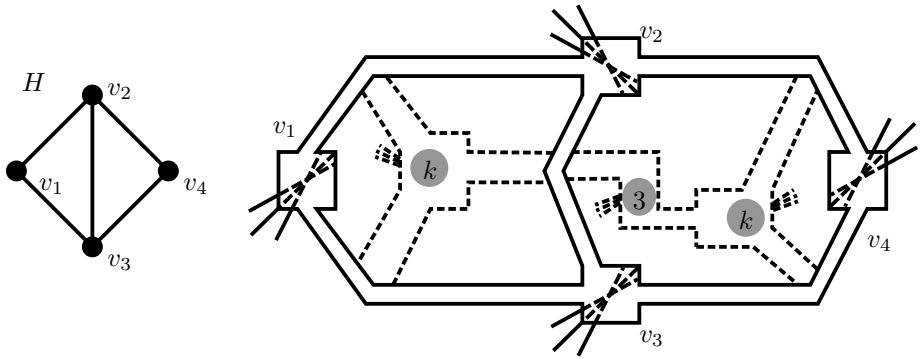
	\end{center}
	\caption{$3$- and $(3+k)$-colorability.
		$3$-colorability is shown in continuous lines, $(3+k)$-colorability is achieved by adding the structures drawn in dashed lines.
		The gray areas contain $k$ or $3$ forced guards.}
	\label{fig:3color}
\end{figure}

\section{2-Colorability}
\label{sec:two-colorability}
{\small
\begin{figure*}
  \subfigure[Variable gadget.]{
    \def\svgwidth{.60\textwidth}
    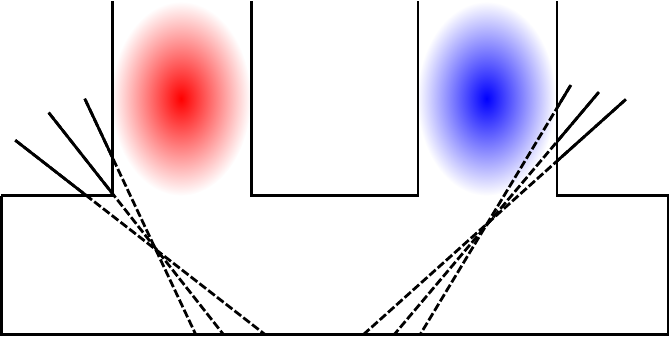
    \label{fig:gadget-variable}
  }
  \hfill
  \subfigure[Inverter gadget.]{
    \def\svgwidth{.09\textwidth}
    \hspace{7mm}
    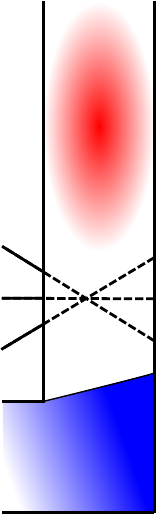
    \hspace{7mm}
    \label{fig:gadget-inverter}
  }
  \subfigure[Crossing gadget. $\V(w_1) \cap \V(w_2)$ indicated in gray. ]{
    \def\svgwidth{.60\textwidth}
    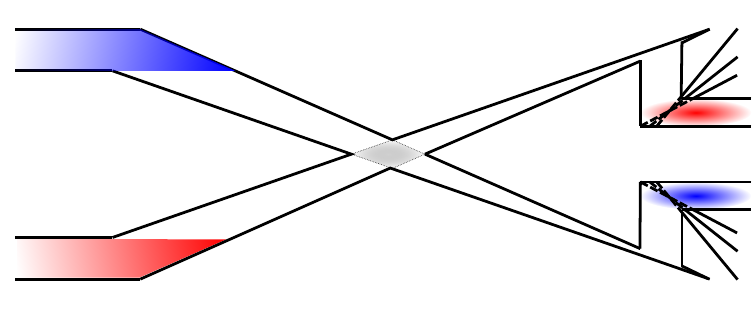
    \label{fig:gadget-crossing}
  }
  \hfill
  \subfigure[And gadget.]{
    \def\svgwidth{.15\textwidth}
    \hspace{5mm}
    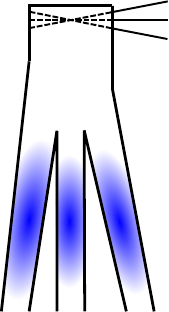
    \label{fig:gadget-and}
  }
  \subfigure[Multiplexer gadget.]{
    \def\svgwidth{.20\textwidth}
    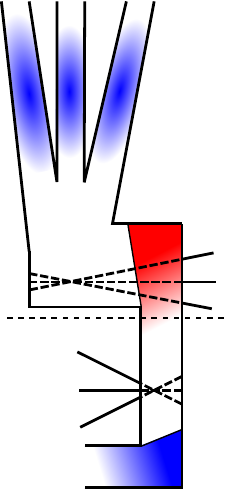
    \hspace{5mm}
    \label{fig:gadget-multiplexer}
  }
  \hfill
  \subfigure[Or gadget, mixed input.]{
    \def\svgwidth{.20\textwidth}
    \hspace{5mm}
    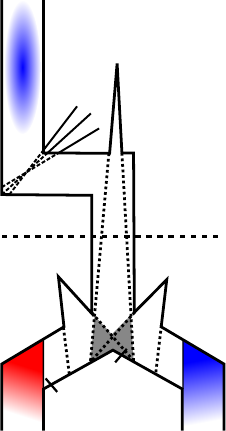
    \hspace{5mm}
    \label{fig:gadget-or-rb}
  }
  \hfill
  \subfigure[Or gadget, uniform input.]{
    \def\svgwidth{.20\textwidth}
    \hspace{5mm}
    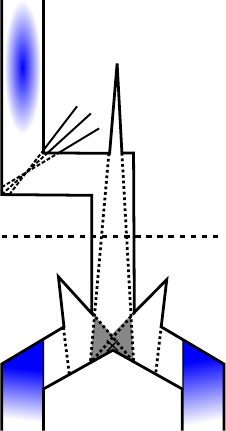
    \hspace{5mm}
    \label{fig:gadget-or-bb}
  }
  \caption{\textsc{3SAT} to $2$-CAGP reduction gadgets.}
  \label{fig:gadgets}
\end{figure*}
}

We use a reduction from \textsc{3SAT} for showing \NP-hardness of deciding whether $\chi_G(P)=2$.
Throughout this section, we associate colors with Boolean values;
\obda, blue corresponds to {\em true} and red to {\em false}.
The output of every gadget is (according to Lemma~\ref{lem:needles}) a guard at a specific position, colored in red or blue which entirely covers all output tunnels serving as input for the next gadgets.

{\bf Variable Gadget.}
This gadget uses the construction in Figure~\ref{fig:gadget-variable} to encode one decision variable $x_i$.
The needles enforce locating two guards at the indicated positions.
The color used for the left guard is interpreted as the value of $x_i$:
Blue means {\em true}, red means {\em false}.

{\bf Inverter Gadget.}
\label{sec:inverter-gadget}
The gadget in Figure~\ref{fig:gadget-inverter} inverts colors.
Its input area is illuminated by one color; the guard forced to
position $g$ must have the other, or a point in the lower right corner can
observe two guards of the same color.

{\bf Crossing Gadget.}
Crossings of channels propagating colors is achieved by the
gadget in Figure~\ref{fig:gadget-crossing}.
For any guard $g$ that sees $w_1$,
$\V(w_1) \subseteq \V(g)$. As $\V(w_1)$ intersects
both the input area and $\V(g_1)$, $g_1$ must have
the same color as the corresponding input; the same holds for $w_2$.
If both input areas are covered by guards of the same color, then
one guard of the opposite color is placed in $\V(w_1) \cap \V(w_2)$.
Otherwise, we place two guards of different color outside of
$\V(w_1) \cap \V(w_2)$, \eg, at $w_1$ and $w_2$.

\ignore{
Crossings of channels propagating colors is achieved by the
gadget in Figure~\ref{fig:gadget-crossing}.
If both input areas are covered by guards of the same color,
\obda blue, no blue guard can be placed in $\V(w_1) \cup \V(w_2)$,
so covering $w_3$ and $w_4$ is only possible using exactly one
red guard in $\V(w_1) \cap \V(w_2)$.
If the inputs are blue and red, no guard can be placed in
$\V(w_1) \cap \V(w_2)$.
One guard is needed in each corridor, \eg, near $w_3$ and $w_4$.
In any case, opposite input colors are propagated to the
dashed line; attaching inverters provides a complete crossing gadget.
}

{\bf Multiplexer Gadget.}
Multiplexing is achieved with the gadget in Figure~\ref{fig:gadget-multiplexer}.
It uses an inverter gadget, and forces a guard to position $g_2$ which covers all output tunnels.
The gadget is easily generalized to an arbitrary number of output tunnels.

{\bf Or Gadget.}
The gadget in Figures~\ref{fig:gadget-or-rb}--\ref{fig:gadget-or-bb} is a binary {\em or}, allowing easy construction of a ternary one.
We argue that there exists a guard cover that colors the output area blue, iff at least one input area is blue.

If two different input colors are applied, a guard cover with blue output exists:
$g_1$ blue, $g_2$ red, and $g_3$ blue in Figure~\ref{fig:gadget-or-rb}.
The same is true for blue/blue input:
$g_1$ red and $g_3$ blue in Figure~\ref{fig:gadget-or-bb}.

If the input is red/red, the output cannot be blue:
By Lemma~\ref{lem:needles}, the gray area must contain a guard $g$,
which can only be blue.
$\V(g) \cap \V(g_3) \neq \emptyset$, so $g_3$ is red.

{\bf And Gadget.}
This gadget is similar to the multiplexer gadget, see Figure~\ref{fig:gadget-and}.
The guard forced to position $g$ can be colored iff all input regions have the same color.
Note that this gadget forces all inputs to be identical, either true or false.

{\bf The 3SAT Reduction.}
Any \textsc{3SAT} instance $S = C_1 \land \dots \land C_m$ with variables $x_1, \dots, x_n$ can be encoded as a $2$-CAGP instance using the gadgets and the overall layout depicted in Figures~\ref{fig:gadgets} and~\ref{fig:3sat}.
It is $2$-colorable iff there exists a variable assignment that sets all clauses to true or all clauses to false;
if all clauses are false, we can invert the values of all variables, and obtain a satisfying truth assignment.
Thus, we have established the following theorem.

%

\begin{theorem}
$2$-CAGP, \ie, deciding if a polygon is $2$-colorable, is \NP-hard.
\end{theorem}


\begin{figure}
  \begin{center}
    \def\svgwidth{.9\textwidth}
    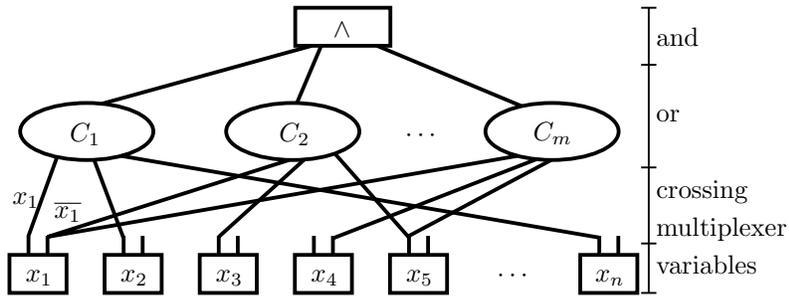
  \end{center}
  \caption{\textsc{3SAT} reduction gadget usage.}
  \label{fig:3sat}
\end{figure}

\section{Simple Polygons}
\label{sec:simple}

Our proof of \NP-hardness of the general CAGP for simple polygons is based on a
reduction from computing a minimum-cardinality set of points for covering a given set of lines;
this auxiliary problem is easily seen to be \NP-hard by geometric duality applied to a result of Megiddo and Tamir~\cite{mt-cllfp-82},
who showed that it is \NP-hard to determine the minimum number of lines to cover a set of points in the plane.

\begin{theorem}
It is \NP-hard to determine the chromatic number of a simple polygon.
\end{theorem}

\begin{figure}
	\begin{center}
                \includegraphics[scale=.35]{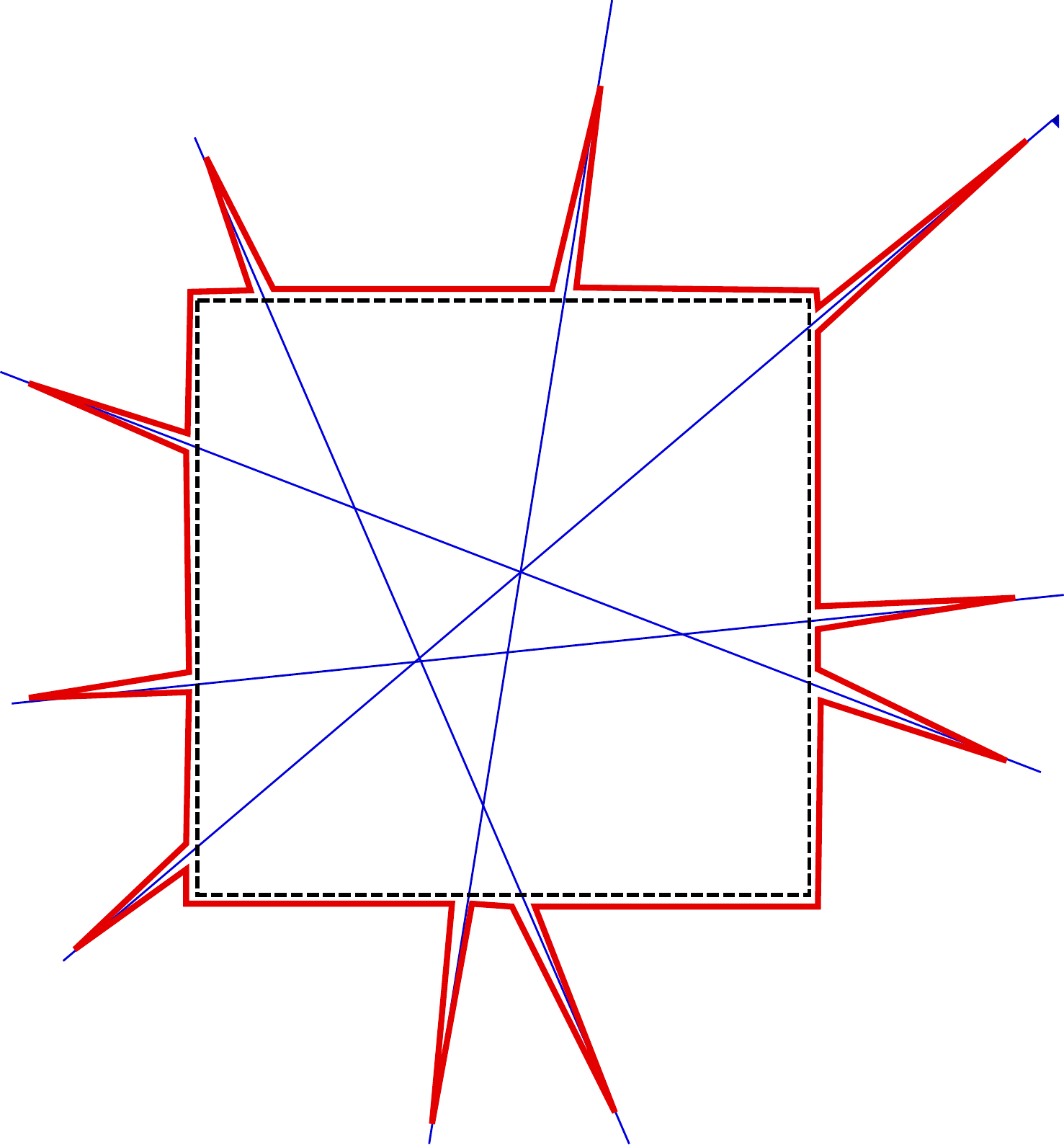}
	\end{center}
	\caption{\NP-hardness of the general CAGP for simple polygons: A minimum-color guard cover of the spike box (red) corresponds to a minimum-cardinality
                 point cover of a set of lines (blue).}
	\label{fig:spike}
\end{figure}

\begin{proof}
Refer to Figure~\ref{fig:spike}.
For a given set of lines, construct a ``spike box'', which is formed by a square that
contains all intersection points of the lines, and has two narrow niche extensions for each line, one at either intersection
with the square. Note that the visibility regions of any two points in the spike box overlap; thus, minimizing the number of colors is equivalent to
minimizing the number of guards.
Now any guard cover corresponds to a point cover of the lines; conversely, any
line cover can be converted to a guard cover of the spike box: if there is a guard placed in a niche, replace it by
one inside the square.
\end{proof}

\section{Conclusion}
\label{sec:conclusion}

A number of open problems remain. These include the complexity of the CAGP for simple polygons and fixed $\chi_G(P)=k$, in particular, $k=2$.
This remains open even for a finite set of candidate locations, as the \NP-hardness proof by Erickson and LaValle~\cite{el-hmlcn-11} requires large $\chi_G(P)$.
Among the other open questions for fixed guard locations is the complexity of determining the chromatic number of a given guard set in a simple
polygon; the claim
by Erickson and LaValle stated in~\cite{el-hmlcn-11} that this problem has a polynomial solution is still unproven, as the corresponding
conflict graph does {\em not} have to be chordal: this can be seen from the n/2 blue and green locations at the spikes in Figure~\ref{fig:example}.

\section*{Acknowledgement}

Michael Hemmer received support by the DFG, contract KR~3133/1-1 (Kunst!).

{\small
\bibliographystyle{abbrv}
\bibliography{bibliography}
}

\end{document}